\documentclass[10pt]{amsart}
\usepackage{bbm,amsfonts,latexsym,amsmath,amscd,amssymb,fancyhdr}
\usepackage[all]{xy}

\theoremstyle{plain}
\newtheorem{theorem}{Theorem}
\numberwithin{theorem}{subsection}

\newtheorem{corollary}{Corollary}
\numberwithin{corollary}{subsection}

\newtheorem{definition}{Definition}
\numberwithin{definition}{section}

\newtheorem{lemma}{Lemma}
\numberwithin{lemma}{section}

\newtheorem{proposition}{Proposition}
\numberwithin{proposition}{section}

\numberwithin{remark}{section}

\numberwithin{equation}{section}

\newcommand {\be}{\begin{equation}}
\newcommand {\ee}{\end{equation}}

\newcommand{\h}{\begin{eqnarray*}}
 \newcommand{\e}{\end{eqnarray*}}

\newcommand{\CC}{\mathbf{C}}

\newcommand{\M}{\rightarrow}

\begin{document}
\title[Gravitational Anomaly Cancellation and Modular Invariance]{Gravitational Anomaly
 Cancellation and Modular Invariance}
\author{Fei Han}
\address{F. Han, Department of Mathematics, National University of Singapore,
 Block S17, 10 Lower Kent Ridge Road,
Singapore 119076 (mathanf@nus.edu.sg)}
\author{Kefeng Liu}
\address{K. Liu, Department of Mathematics, University of California at Los Angeles,
Los Angeles, CA 90095, USA (liu@math.ucla.edu) and  Center of Mathematical Sciences, Zhejiang University, 310027, P.R. China}
\maketitle

\begin{abstract}  In this paper, by combining modular forms and characteristic forms,
we obtain general anomaly cancellation formulas of any dimension.
For $4k+2$ dimensional manifolds, our results include the
gravitational anomaly cancellation formulas of
Alvarez-Gaum\'e and Witten in dimensions 2, 6 and 10 (\cite{AW}) as
special cases.  In dimension $4k+1$, we derive anomaly cancellation
formulas for index gerbes. In dimension $4k+3$, we obtain certain
results about eta invariants, which are interesting in spectral
geometry.
\end{abstract}

\section{Introduction}
In \cite{AW},  it is shown that in certain parity-violating gravity theory in $4k+2$ dimensions,
when  Weyl fermions of spin-$\frac{1}{2}$ or spin-$\frac{3}{2}$ or self-dual antisymmetric tensor
 field are coupled to gravity, perturbative anomalies occur.  Alvarez-Gaum\'e and Witten calculate
 the anomalies and  show that there are cancellation formulas for these anomalies in dimensions
 $2, 6, 10$.  Let $\widehat{I}_{1/2}, \widehat{I}_{3/2}$ and $\widehat{I}_{A}$ be the spin-$\frac{1}{2}$,
   spin-$\frac{3}{2}$ and antisymmetric tensor anomalies respectively.  By direct computations,
   Alvarez-Gaum\'e and Witten find anomaly cancellation formulas in dimensions 2, 6, 10  respectively,
\be -\widehat{I}_{1/2}+\widehat{I}_{A}=0,\ee

\be 21\widehat{I}_{1/2}- \widehat{I}_{3/2}+8\widehat{I}_{A}=0,\ee
and \be -\widehat{I}_{1/2}+\widehat{I}_{3/2}+\widehat{I}_{A}=0.\ee
These anomaly cancellation formulas can tell us how many fermions of
different types should be coupled to the gravity to make the theory
anomaly free.  Alvarez, Singer and Zumion \cite{ASZ} reproduce the above anomalies in a
different way by  using the family index theorem instead of Feynman
diagram methods.

When perturbative anomalies cancel, this means that the effective
action is invariant under gauge and coordinate transformations  that
can be reached continuously from the identity.  In \cite{W}, Witten
introduced the global anomaly by asking whether the effective action
is invariant under  gauge and coordinate transformations  that are
not continuously connected to the identity. Witten's work suggests
that the global anomaly should be related to the holonomy of a
natural connection on the determinant line bundle of the family
Dirac operators.

From the topological point of view, anomaly measures the
nontriviality of the determinant line bundle of a family of Dirac
operators.  The perturbative anomaly detects the real first Chern
class of the determinant line bundle while the global anomaly
detects the integral first Chern class beyond the real information
(cf. \cite{F}).

For a family of Dirac operators on an even dimensional closed
manifold, the determinant line bundle over the parametrizing space
carries the Quillen metric as well as the Bismut-Freed connection
compatible with the Quillen metric such that the curvature of the
Bismut-Freed connection is the two form component of the
Atiyah-Singer family index theorem \cite{BF1, BF2}. The curvature
of the Bismut-Freed connection is the representative (up to a
constant) of the real first Chern class of the determinant line
bundle.   In this paper, by developing modular invariance of certain
characteristic forms, we derive cancellation formulas for the
curvatures of determinant line bundles of family signature operators
and family tangent twisted Dirac operators on $4k+2$ dimensional manifolds
(see Theorem 2.2.1 and Theorem 2.2.2). When $k=0, 1, 2$, i.e. in
dimensions 2, 6, and 10, our cancellation formulas just give the
Alvarez-Gaum\'e-Witten cancellation formulas (1.1)-(1.3) (see
Theorem 2.2.5 and its proof).

For global anomaly, in \cite{BF2}, Bismut and Freed prove the
holonomy theorem suggsted by Witten.
 Later, to detect the integral information of the first Chern class of the determinant line bundle,
  Freed uses Sullivan's $\mathbf{Z}/k$ manifolds \cite{F}. In this paper, we also give cancellation
  formulas for the holonomies (with respect to the Bismut-Freed connections) of
determinant line bundles of family signature operators and family tangent
twisted Dirac operators on $4k+2$ dimensional manifolds for torsion
loops which appear in the data of $\mathbf{Z}/k$ surfaces (Theorem
2.2.3 and 2.2.4).

The general anomaly cancellation formulas in dimension $4k$ have been studied in  \cite{Liu1}.  One naturally asks if there are similar results in odd dimensions.

For a family of Dirac operators on an odd dimensional manifold, Lott
(\cite{Lo}) constructed  an abelian gerbe-with-connection whose
curvature is the three form component of the Atiyah-Singer families
index theorem. This gerbe is called the index gerbe, which is a
higher analogue of the determinant line bundle. As Lott remarks in
his paper that the curvature of such gerbes are also certain
nonabelian gauge anomaly from a Hamitonian point of view. In this paper,  we derive anomaly cancellation formulas for the
curvatures of index gerbes of family odd signature operators and
family tangent twisted Dirac operators on $4k+1$ dimensional manifolds (see 
Theorem 2.3.1, Theorem 2.3.2 and Corollary 2.3.1-2.3.3).  Moreover,
based on a result of Ebert  \cite{Eb}, we can also derive anomaly
cancellation formulas on the de Rham cohomology level (but not on
the form level), which does not involve the family odd signature
operators (see Theorem 2.3.3, Theorem 2.3.4 and Corollary
2.3.1-2.3.3).  We hope there is some physical meaning related to our
cohomological anomaly cancellation formulas.

In dimension $4k+3$, we derive some results for the
reduced $\eta$-invariants of family odd signature operators and family 
tangent twisted Dirac operators,  which are interesting in spectral geometry
(Theorem 2.4.1, Theorem 2.4.2 and Corollary 2.4.1-2.4.3).  Moreover, function of the form $$\exp\{2\pi \sqrt{-1} (\mathrm{linear\ combination\ of\ reduced\ eta\ invariants})\}$$ has appeared in physics (\cite{DMW}) as phase of effective action of $M$-theory  in 11 dimension. 
We hope our results for reduced $\eta$-invariants can also find applications in
physics.

We obtain our anomaly cancellation formulas by combining the family
index theory and modular invariance of characteristic forms.
Gravitational and gauge anomaly cancellations are very important in
physics because they can keep the consistency of certain quantum
field theories. It is quite interesting to notice that these
cancellation formulas are consequences of the modular properties of
characteristic forms which are rooted in elliptic genera.

\section{Results}

In this section, we will first prepare some geometric settings in
Section 2.1 and then present our results in Section 2.2-2.4. The proofs of the theorems in Section 2.2-2.4 will be given in Section 3.

\subsection{Geometric Settings} Following \cite{B1}, we define some geometric data on a fiber bundle as follows. Let $\pi: M \M Y$ be a smooth fiber bundle
with compact fibers $Z$ and connected base $Y$. Let $TZ$ be the
vertical tangent bundle of the fiber bundle and $g^{Z}$ be a metric
on $TZ$. Let $T^HM$ be a smooth subbundle of $TM$ such that $TM=T^HM\oplus TZ$. Assume that
$TY$ is endowed with a metric $g^{Y}$. We lift the metric of $TY$ to
$T^HM$ and by assuming that $T^HM$ and $TZ$ are orthogonal, $TM$ is
endowed with a metric which we denote $g^{Y}\oplus g^{Z}$. Let $\nabla^L$ be the Levi-Civita connection of $TM$ for
the metric $g^{Y}\oplus g^{Z}$ and $P_{Z}$ 
denote the orthogonal projection from $TM$ to $TZ$. Let $\nabla^Z$ denote the connection on $TZ$ defined by the
relation $U\in TM, V\in TZ, \nabla_U^ZV=P_Z\nabla^L_UV$. $\nabla^Z$
preserves the metric $g^Z$. Let $R^Z=\nabla^{Z,2}$ be the curvature of
$\nabla^Z$.

Let $E,F$ be two Hermitian vector bundles over $M$ carrying
Hermitian connections $\nabla^E, \nabla^F$ respectively. Let
$R^E=\nabla^{E,\ 2}$ (resp. $R^F=\nabla^{F,\ 2}$) be the curvature
of $\nabla^E$ (resp. $\nabla^F$). If we set the formal difference
$G=E-F$, then $G$ carries an induced Hermitian connection
$\nabla^G$ in an obvious sense. We define the associated Chern
character form as (cf. \cite{Z})
$$ {\rm ch}(G,\nabla^G)={\rm tr}\left[{\rm exp}\left(\frac{\sqrt{-1}}{2
\pi}R^{E}\right)\right]-{\rm tr}\left[{\rm
exp}\left(\frac{\sqrt{-1}}{2 \pi}R^{F}\right)\right].$$

For any complex number $t$, let
$$\Lambda_t(E)=\mathbf{C}|_M+tE+t^2\Lambda^2(E)+\cdots ,
\ \  S_t(E)=\mathbf{C}|_M+tE+t^2S^2(E)+\cdots$$  denote respectively
the total exterior and symmetric powers  of $E$, which live in
$K(M)[[t]].$ The following relations between these two operations
(\cite{At}, Chap. 3) hold, \be S_t(E)=\frac{1}{\Lambda_{-t}(E)},\ \ \ \
 \Lambda_t(E-F)=\frac{\Lambda_t(E)}{\Lambda_t(F)}.\ee

 The connections
$\nabla^E, \nabla^F$ naturally induce connections on $S_tE,
\Lambda_tE,$ etc. Moreover, if $\{\omega_i \}$, $\{{\omega_j}' \}$
are formal Chern roots for Hermitian vector bundles $E$, $F$
respectively, then [15, Chap. 1] \be
\mathrm{ch}\left(\Lambda_t{(E)},
\nabla^{\Lambda_t(E)}\right)=\prod\limits_i(1+e^{\omega_i}t).\ee

We have the following formulas for Chern character forms,
\be{\rm ch}\left(S_t(E), \nabla^{S_t(E)} \right)=\frac{1}{{\rm
ch}\left(\Lambda_{-t}(E),\nabla^{\Lambda_{-t}(E)}
\right)}=\frac{1}{\prod\limits_i (1-e^{\omega_i}t)}\ ,\ee \be{\rm
ch}\left(\Lambda_t(E-F), \nabla^{\Lambda_t(E-F)} \right)=\frac{{\rm
ch}\left(\Lambda_{t}(E),\nabla^{\Lambda_t(E)} \right)}{{\rm
ch}\left(\Lambda_{t}(F),\nabla^{\Lambda_t(F)}
\right)}=\frac{\prod\limits_i(1+e^{\omega_i}t)}{\prod\limits_j(1+e^{{\omega_j}'}t)}\
.\ee

If $W$ is a  real Euclidean vector bundle over $M$ carrying a
Euclidean connection $\nabla^W$, then its complexification
$W_\CC=W\otimes \CC$ is a complex vector bundle over
$M$ carrying a canonically induced Hermitian metric from that of
$W$, as well as a Hermitian connection $\nabla^{W_\CC}$
induced from $\nabla^W$.

If $E$ is a vector bundle (complex or real)
over $M$, set $\widetilde{E}=E-{{\rm dim}E}$ in $K(M)$ or $KO(M)$.

Let $q=e^{2\pi \sqrt{-1}\tau}$ with $\tau \in \mathbf{H}$, the upper
half complex plane. Let $T_\CC Z$ be the complexification of $TZ$. Set
\be \Theta_1(T_\CC Z)=\bigotimes_{n=1}^\infty S_{q^n}(\widetilde{T_\CC Z})
\otimes \bigotimes_{m=1}^\infty \Lambda_{q^m}(\widetilde{T_\CC Z}), \ee

\be\Theta_2(T_\CC Z)=\bigotimes_{n=1}^\infty
S_{q^n}(\widetilde{T_\CC Z}) \otimes \bigotimes_{m=1}^\infty
\Lambda_{-q^{m-{1\over 2}}}(\widetilde{T_\CC Z}),\ee which are
elements in $K(M)[[q^{1\over2}]]$.

$\Theta_1(T_\CC Z)$ and $\Theta_2(T_\CC Z)$ admit
formal Fourier expansion in $q^{1/2}$ as
\be \Theta_1(T_\CC Z)=A_0(T_\CC Z)+ A_1(T_\CC Z)q^{1/2}+\cdots, \ee
\be \Theta_2(T_\CC Z)=B_0(T_\CC Z)+B_1(T_\CC Z)q^{1/2}+\cdots,\ee
where the $A_j$'s and $B_j$'s are elements in the semi-group formally generated by
complex vector bundles over $M$. Moreover, they carry canonically induced
connections denoted by $\nabla^{A_j}$ and $\nabla^{B_j}$ respectively, and let
$\nabla^{\Theta_1(T_\CC Z)}$, $\nabla^{\Theta_2(T_\CC Z)}$ be the induced connections
with $q^{1/2}$-coefficients on
$\Theta_1$, $\Theta_2$ from the $\nabla^{A_j}$, $\nabla^{B_j}$.

The four Jacobi theta functions are defined as follows (cf.
\cite{C}): \be\theta(v,\tau)=2q^{1/8}\sin(\pi v)
\prod_{j=1}^\infty\left[(1-q^j)(1-e^{2\pi \sqrt{-1}v}q^j)(1-e^{-2\pi
\sqrt{-1}v}q^j)\right]\ ,\ee \be \theta_1(v,\tau)=2q^{1/8}\cos(\pi
v)
 \prod_{j=1}^\infty\left[(1-q^j)(1+e^{2\pi \sqrt{-1}v}q^j)
 (1+e^{-2\pi \sqrt{-1}v}q^j)\right]\ ,\ee
\be \theta_2(v,\tau)=\prod_{j=1}^\infty\left[(1-q^j)
 (1-e^{2\pi \sqrt{-1}v}q^{j-1/2})(1-e^{-2\pi \sqrt{-1}v}q^{j-1/2})\right]\
 ,\ee
\be \theta_3(v,\tau)=\prod_{j=1}^\infty\left[(1-q^j) (1+e^{2\pi
\sqrt{-1}v}q^{j-1/2})(1+e^{-2\pi \sqrt{-1}v}q^{j-1/2})\right]\ .\ee
They are all holomorphic functions for $(v,\tau)\in \mathbf{C \times
H}$, where $\mathbf{C}$ is the complex plane and $\mathbf{H}$ is the
upper half plane.

Define two $q$-series (see Section 3 for details) \be
\delta_2(\tau)=-\frac{1}{8}(\theta_1(0,\tau)^4+\theta_3(0,\tau)^4),
\ \ \ \ \varepsilon_2(\tau)=\frac{1}{16}\theta_1(0,\tau)^4
\theta_3(0,\tau)^4.\ee  They have the following Fourier expansions
in $q^{1/2}$: $$\delta_2(\tau)=-{1\over 8}-3q^{1/2}-3q+\cdots,\ \ \
\ \varepsilon_2(\tau)=q^{1/2}+8q+\cdots.$$

When the dimension of the fiber is $8m+1, 8m+2$ or $8m+3$, define
virtual complex vector bundles $b_r(T_\CC Z)$ on $M$, $0\leq r \leq
m,$ via the equality \be \Theta_2(T_\CC Z)\equiv \sum_{r=0}^m
b_r(T_\CC Z)(8\delta_2)^{2m+1-2r}\varepsilon_2^r\  \ \ \
\mathrm{mod}\, q^{\frac{m+1}{2}}\cdot K(M)[[q^{\frac{1}{2}}]].
\ee

When the dimension of the fiber is $8m-1, 8m-2$ or $8m-3$, define
virtual complex vector bundles $z_r(T_\CC Z)$ on $M$, $0\leq r \leq
m,$ via the equality \be \Theta_2(T_\CC Z)\equiv \sum_{r=0}^m
z_r(T_\CC Z)(8\delta_2)^{2m-2r}\varepsilon_2^r\ \ \ \ \
\mathrm{mod}\, q^{\frac{m+1}{2}}\cdot
K(M)[[q^{\frac{1}{2}}]].\ee

It's not hard to show that each $b_r(T_\CC Z), 0\leq r \leq m$, is a
canonical linear combination of $B_j(T_\CC Z), 0\leq j\leq r.$ This
is also true for $z_r(T_\CC Z)$'s. These $b_r(T_\CC Z)$'s and
$z_r(T_\CC Z)$'s carry canonically induced metrics and connections.

From (2.14) and (2.15), it's not hard to calculate that
\be b_0(T_\CC Z)=-\mathbf{C}, \ b_1(T_\CC Z)=T_\CC Z+\mathbf{C}^{24(2m+1)-\mathrm{dim}Z}\ee and
\be z_0(T_\CC Z)=\mathbf{C},\ z_1(T_\CC Z)=-T_\CC Z-\mathbf{C}^{48m-\mathrm{dim}Z}.\ee

\subsection{Determinant Line Bundles and Anomaly Cancellation Formulas}

Suppose the dimension of the fiber is $2n$ and the dimension of the
base $Y$ is $p$. Assume that $TZ$ is oriented. Let $T^*Z$ be the
dual bundle of $TZ$.

Let $E=\overset{2n}{\underset{i=0}{\oplus}} E^i$ be the smooth
infinite-dimensional $\mathbf{Z}$-graded vector bundle over $Y$
whose fibre over $y\in Y$ is $C^\infty(Z_y, \Lambda_\CC
(T^*Z)|_{Z_y})$, i.e.
$$C^\infty(Y, E^i)=C^\infty(M, \Lambda_\CC (T^*Z)), $$ where $\Lambda_\CC (T^*Z)$ is the complexified
exterior algebra bundle of $TZ$.

For $X\in TZ$, let $c(X), \widehat{c}(X)$ be the Clifford actions on
$\Lambda_\CC (T^*Z)$ defined by $c(X)=X^*-i_X, \widehat{c}(X)=X^*+i_X$, where $X^*\in T^*Z$ corresponds to $X$ via
$g^Z$.

Let $\{e_1,e_2, \cdots, e_{2n}\}$ be an oriented orthogonal basis of $TZ$. Set
$$\Omega=(\sqrt{-1})^{n}c(e_1)\cdots c(e_{2n}).$$ Then $\Omega$ is a self-adjoint element acting
on $\Lambda_\CC (T^*Z)$ such that $\Omega^2=\mathrm{Id}|_{\Lambda_\CC(T^*Z)}$.

Let $dv_Z$  be the Riemannian volume form on fibers $Z$ associated
to the metric $g^{Z}$ ($dv_Z$ is actually a section of
$\Lambda_\CC^{\mathrm{dim}Z}(T^*Z)$). Let $\langle \ , \
\rangle_{\Lambda_\CC (T^*Z)}$ be metric on $\Lambda_\CC(T^*Z)$
induced by $g^{Z}$. Then $E$ has a Hermitian metric $h^E$ such that
for $\alpha, \alpha'\in C^\infty(Y, E)$ and $y\in Y$,
$$\langle \alpha, \alpha'\rangle_{h^E}(y)=\int_{Z_y}\langle \alpha,
 \alpha'\rangle_{\Lambda_\CC (T^*Z)}\,dv_{Z_y}.$$

Let $d^Z$ be the exterior differentiation along fibers. $d^Z$ can be
considered as an element of $C^\infty(Y, \mathrm{Hom}(E^\bullet,
E^{\bullet+1}))$. Let $d^{Z\ast}$ be the formal adjoint of $d^Z$
with respect to the inner product $\langle\ ,\ \rangle_{h^E}$.
Define the family signature operator (c.f.\cite{MZ1}) $D^Z_{sig}$ to be
\be D^Z_{sig}=d^Z+d^{Z\ast}: C^\infty(M, \Lambda_\CC(T^*Z))\M
C^\infty(M, \Lambda_\CC(T^*Z)).\ee The $\mathbf{Z}_2$-grading of
$D^Z_{sig}$ is given by the $+1$ and $-1$ eigenbundles of $\Omega$.
Clearly, for each $y\in Y$,  $$(D^Z_{sig})_y: C^\infty(Z_y,
\Lambda_\CC(T^*Z)|_y)\M C^\infty(Z_y, \Lambda_\CC(T^*Z)|_y)$$ is the
signature operator for the fiber $Z_y$.

Further assume that $TZ$ is spin. Following \cite{B1}, the family Dirac operators are defined as follows.

Let $O$ be the $SO(2n)$ bundle of oriented orthogonal frames in
$TZ$. Since $TZ$ is spin, the $SO(2n)$ bundle
$\xymatrix@C=0.5cm{O\ar[r]_{\varrho}& M}$ lifts to a $Spin(2n)$
bundle $$\xymatrix@C=0.5cm{O'\ar[r]_{\sigma}& O\ar[r]_{\varrho}&
M}$$ such that $\sigma$ induces the covering projection $Spin(2n)\M
SO(2n)$ on each fiber. Assume $F, F_{\pm}$ denote the Hermitian
bundles of spinors
$$F=O'\times_{Spin(2n)}S_{2n}, \ F_{\pm}=O'\times_{Spin(2n)}S_{\pm, {2n}},$$ where
$S_{2n}=S_{+, {2n}}\oplus S_{-, {2n}}$ is the space of complex
spinors. The connection $\nabla^Z$ on $O$ lifts to a connection on
$O'$. $F, F_{\pm}$ are then naturally endowed with a unitary
connection, which we simply denote by $\nabla$.

Let $V$ be a $l$-dimensional complex Hermitian bundle on $M$. Assume
that $V$ is endowed with a unitary connection $\nabla^V$ whose
curvature is $R^V$. The Hermitian bundle $F\otimes V$ is naturally
endowed with a unitary connection which we still denote by $\nabla$.

Let $H^\infty, H_{\pm}^\infty$ be the sets of $C^\infty$ sections of $F\otimes V, F_{\pm}\otimes V$ over $M$. $H^\infty, H^\infty_{\pm}$ are viewed as the sets of $C^\infty$ sections over $Y$ of infinite dimensional bundles which are still denoted by $H^\infty, H^\infty_{\pm}$. For $y\in Y$, $H^\infty_y, H^\infty_{y, \pm}$ are the sets of $C^\infty$ sections over $Z_y$ of $F\otimes V, F_{\pm}\otimes V$.

The elements of $TZ$ acts by Clifford multiplication on $F\otimes
V$. Suppose $\{e_1, e_2,\cdots e_{2n}\}$ is a local orthogonal basis
of $TZ$. Define the family Dirac operator twisted by $V$ to be
$D^{Z}\otimes V=\overset{2n}{\underset{i=1}{\sum}}
e_i\nabla_{e_i}$. Let $(D^{Z}\otimes V)_{\pm}$ denote the
restriction of $D^{Z}\otimes V$ to $H^\infty_{\pm}$. For each $y\in
Y$, \be (D^{Z}\otimes V)_y=\left[\begin{array}{cc}
      0&(D^{Z}\otimes V)_y,_-\\
      (D^{Z}\otimes V)_y,_+&0
\end{array}\right]\in \mathrm{End}^{odd}(H^\infty_{y, +}\oplus H^\infty_{y, -})\ee is the twisted
Dirac operator on the fiber $Z_y$.

The family signature operator is a twisted
family Dirac operator. Actually, we have $D^Z_{sig}=D^{Z}\otimes F$
(cf. \cite{BGV}).

Let $\mathcal{L}_{D^{Z}\otimes
V}=\mathrm{det}(\mathrm{Ker}(D^{Z}\otimes V)_+)^*\otimes
 \mathrm{det}(\mathrm{Ker}(D^{Z}\otimes V)_-)$ be the determinant
line bundle of the family operator $D^{Z}\otimes V$ over $Y$
(\cite{Q, BF1}). The nontriviality of $\mathcal{L}_{D^{Z}\otimes V}$
is certain anomaly in physics. 

The determinant line bundle carries the Quillen metric
$g^{\mathcal{L}_{D^{Z}\otimes V}}$ as well as the Bismut-Freed
connection $\nabla^{\mathcal{L}_{D^{Z}\otimes V}}$ compatible to
$g^{\mathcal{L}_{D^{Z}\otimes V}}$,  the curvature
$R^{\mathcal{L}_{D^{Z}\otimes V}}$ of which is equal to the two-form
component of the Atiyah-Singer families index theorem \cite{BF1,
BF2}. $\frac{1}{2\pi}R^{\mathcal{L}_{D^{Z}\otimes V}}$ is a
representative of the local anomaly.

For the global anomaly, in \cite{BF2}, Bismut and Freed give a
heat equation proof of the holonomy theorem in the form suggested by
Witten in \cite{W}. To detect information for the integral first
Chern class of $\mathcal{L}_{D^{Z}\otimes V}$,  Freed uses
$\mathbf{Z}/k$ manifolds in \cite{F}. $\mathbf{Z}/k$ manifold is introduced
by Sullivan in his studies of geometric topology.  A closed
$\mathbf{Z}/k$ manifold (c.f. \cite{F})  consists of (1)   a compact
manifold $Q$ with boundary; (2)  a closed manifold $P$; (3) a
decomposition $\partial
Q={\underset{i=1}{\overset{k}\coprod}}(\partial Q)_i$ of the
boundary of $Q$ into $k$ disjoint manifolds and diffeomorphisms
$\alpha_i: P\rightarrow (\partial Q)_i$. The  identification space
$\overline{Q}$, formed by attaching $Q$ to $P$ by $\alpha_i$ is more
properly called $\mathbf{Z}/k$ manifolds.  $\overline{Q}$ is
singular at identification points.  If $Q$ and $P$ are compatibly
oriented, then  $\overline{Q}$ carries a fundamental class
$[\overline{Q}]\in H_*(\overline{Q}, \mathbf{Z}/k)$. In \cite{F}, the first Chern class of the determinant line bundle over
$\overline{\Sigma}\rightarrow Y$ is evaluated  for all $\mathrm{Z}/k$ surfaces and
all maps to detect the rest information other than the real
information.

For local anomalies, we have the following cancellation formula formulas.
\begin{theorem}If the fiber is $8m+2$ dimensional, then the following local
anomaly cancellation formula holds,
\be
R^{\mathcal{L}_{D^Z_{sig}}}-8\sum_{r=0}^{m}2^{6m-6r}R^{\mathcal{L}_{D^Z\otimes
b_r(T_\CC Z)}}=0.\ee \end{theorem}

\begin{theorem} If the fiber be $8m-2$ dimensional, then the following local anomaly
 cancellation formula holds,
\be
R^{\mathcal{L}_{D^Z_{sig}}}-\sum_{r=0}^{m}2^{6m-6r}R^{\mathcal{L}_{D^Z\otimes
z_r(T_\CC Z)}}=0.\ee \end{theorem}

For global anomalies, we have the following cancellation formulas concerning the holonomies.

\begin{theorem} If the fiber is $8m+2$ dimensional, $(\Sigma, S)$ is a $\mathbf{Z}/k$ surface and
$f:\overline{\Sigma}\to Y$ is a map, then
\be
\begin{split}& \frac{\sqrt{-1}}{2\pi}\mathrm{ln
hol}_{\mathcal{L}_{D^Z_{sig}}}(S)-8\sum_{r=0}^{m}2^{6m-6r}\frac{\sqrt{-1}}{2\pi}\mathrm{ln
hol}_{\mathcal{L}_{D^Z\otimes
b_r(T_\CC Z)}}(S)\\
\equiv & c_1(f^*(\mathcal{L}_{D^Z_{sig}}))[\overline{\Sigma}]
-8\sum_{r=0}^{m}2^{6m-6r}c_1\left(f^*\left(\mathcal{L}_{D^Z\otimes
b_r(T_\CC Z)}\right)\right)[\overline{\Sigma}]\ \
\mathrm{mod}\  1,\end{split}\ee where we view
$\mathbf{Z}/k\cong\mathbf{Z}\left[1/k \right]/\mathbf{Z}\subset
\mathbf{Q}/\mathbf{Z}.$
\end{theorem}

\begin{theorem} If the fiber is $8m-2$ dimensional, $(\Sigma, S)$ is a $\mathbf{Z}/k$ surface and
 $f:\overline{\Sigma}\to Y$ is a map, then
\be
\begin{split}& \frac{\sqrt{-1}}{2\pi}\mathrm{ln
hol}_{\mathcal{L}_{D^Z_{sig}}}(S)-\sum_{r=0}^{m}2^{6m-6r}\frac{\sqrt{-1}}{2\pi}\mathrm{ln
hol}_{\mathcal{L}_{D^Z\otimes
z_r(T_\CC Z)}}(S)\\
\equiv & c_1(f^*(\mathcal{L}_{D^Z_{sig}}))[\overline{\Sigma}]
-\sum_{r=0}^{m}2^{6m-6r}c_1\left(f^*\left(\mathcal{L}_{D^Z\otimes
z_r(T_\CC Z)}\right)\right)[\overline{\Sigma}]\ \
\mathrm{mod}\  1,\end{split}\ee where we view
$\mathbf{Z}/k\cong\mathbf{Z}\left[1/k \right]/\mathbf{Z}\subset
\mathbf{Q}/\mathbf{Z}.$
\end{theorem}

Putting $m=0$ and $m=1$ in the above theorems  and using (2.16) as well as (2.17), we have
\begin{corollary}If the fiber is $2$ dimensional, then the
following local anomaly cancellation formula holds,
\be
R^{\mathcal{L}_{D^Z_{sig}}}+8R^{\mathcal{L}_{D^Z}}=0.\ee If $(\Sigma, S)$ is a $\mathbf{Z}/k$ surface and
 $f:\overline{\Sigma}\to Y$ is a map, then \be
\begin{split}&\frac{\sqrt{-1}}{2\pi}\mathrm{ln
hol}_{\mathcal{L}_{D^Z_{sig}}}(S)+8\frac{\sqrt{-1}}{2\pi}\mathrm{ln
hol}_{\mathcal{L}_{D^Z}}(S)\\
\equiv &c_1\left(f^*(\mathcal{L}_{D^Z_{sig}})\right)[\overline{\Sigma}]
+8c_1\left(f^*\left(\mathcal{L}_{D^Z}\right)\right)[\overline{\Sigma}] \ \ \ \ \ \ \ \ \
\mathrm{mod}\  1.\end{split}\ee
\end{corollary}

\begin{corollary}If the fiber is $6$ dimensional, then the following local anomaly cancellation formula
holds,
\be
R^{\mathcal{L}_{D^Z_{sig}}}+R^{\mathcal{L}_{D^Z\otimes T_\CC Z}}-22R^{\mathcal{L}_{D^Z}}=0.\ee
 If $(\Sigma, S)$ is a $\mathbf{Z}/k$ surface and  $f:\overline{\Sigma}\to Y$ is a map, then \be
\begin{split}&\frac{\sqrt{-1}}{2\pi}\mathrm{ln
hol}_{\mathcal{L}_{D^Z_{sig}}}(S)+\frac{\sqrt{-1}}{2\pi}\mathrm{ln
hol}_{\mathcal{L}_{D^Z\otimes T_\CC Z}}(S)-22\frac{\sqrt{-1}}{2\pi}\mathrm{ln
hol}_{\mathcal{L}_{D^Z}}(S) \\
\equiv & c_1(f^*(\mathcal{L}_{D^Z_{sig}}))[\overline{\Sigma}]
+c_1\left(f^*\left(\mathcal{L}_{D^Z\otimes T_\CC Z}\right)\right)[\overline{\Sigma}]-22c_1
\left(f^*\left(\mathcal{L}_{D^Z}\right)\right)[\overline{\Sigma}]\ \
\mathrm{mod}\  1.\end{split}\ee
\end{corollary}

\begin{corollary}If the fiber is $10$ dimensional, then the the following local anomaly
cancellation formula holds,
\be
R^{\mathcal{L}_{D^Z_{sig}}}-8R^{\mathcal{L}_{D^Z\otimes T_\CC Z}}+16R^{\mathcal{L}_{D^Z}}=0.\ee
 If $(\Sigma, S)$ is a $\mathbf{Z}/k$ surface and  $f:\overline{\Sigma}\to Y$ is a map, then \be
\begin{split}&\frac{\sqrt{-1}}{2\pi}\mathrm{ln
hol}_{\mathcal{L}_{D^Z_{sig}}}(S)-8\frac{\sqrt{-1}}{2\pi}\mathrm{ln
hol}_{\mathcal{L}_{D^Z\otimes T_\CC Z}}(S)+16\frac{\sqrt{-1}}{2\pi}\mathrm{ln
hol}_{\mathcal{L}_{D^Z}}(S) \\
\equiv & c_1(f^*(\mathcal{L}_{D^Z_{sig}}))[\overline{\Sigma}]
-8c_1\left(f^*\left(\mathcal{L}_{D^Z\otimes T_\CC Z}\right)\right)[\overline{\Sigma}]+
16c_1\left(f^*\left(\mathcal{L}_{D^Z}\right)\right)[\overline{\Sigma}]\ \
\mathrm{mod}\  1.\end{split}\ee

\end{corollary}

Our anomaly cancellation formulas actually imply the Alvarez-Gaum\'e
and Witten anomaly cancellation formulas.
\begin{theorem} In dimensions $2, 6, 10$, our anomaly cancellation formulas (2.24),
(2.26) and (2.28) give the gravitational anomaly cancellation formulas (1.1)-(1.3) of
Alvarez-Gaum\'e and Witten.\end{theorem}
The proof this theorem will also be given in Section 3.

\subsection{Index Gerbes and Anomaly Cancellation Formulas}
Now we still assume that $TZ$ is oriented but the dimension of the
fiber is $2n+1$, i.e. we consider odd dimensional fibers. We still
adopt the geometric settings in Section 2.1.

Let $\{e_1,e_2, \cdots, e_{2n+1}\}$ be an oriented orthogonal basis of $TZ$. Set
$$\Gamma=(\sqrt{-1})^{n+1}c(e_1)\cdots c(e_{2n+1}).$$ Then $\Gamma$ is a self-adjoint element acting
on $\Lambda_\CC(T^*Z)$ such that $\Gamma^2=\mathrm{Id}|_{\Lambda_\CC(T^*Z)}$.

Define the family odd signature operator $B^Z_{sig}$ to be \be
B^Z_{sig}=\Gamma d^Z+d^Z\Gamma: C^\infty(M, \Lambda_\CC
^{even}(T^*Z))\M C^\infty(M, \Lambda_\CC^{even}(T^*Z)).\ee  For each
$y\in Y$,  \be (B^Z_{sig})_y: C^\infty(Z_y,
\Lambda_\CC^{even}(T^*Z)|_y)\M C^\infty(Z_y, \Lambda_\CC
^{even}(T^*Z)|_y)\ee is the odd signature operator $B_{even}$ for the
fiber $Z_y$ in \cite{APS2}.

Now assume that $TZ$ is spin and still let $V$ be a $l$-dimensional
complex Hermitian bundle with the unitary connection $\nabla^V$. One
can still define the family Dirac operator $D^{Z}\otimes V$ similar
as the even dimensional fiber case. The only difference is that now
the spinor bundle $F'$ associated to $TZ$ is not
$\mathbf{Z}_2$-graded. Let $H^\infty$ be the set of $C^\infty$
sections of $F'\otimes V$ over $M$. $H^\infty$ is viewed as the set
of $C^\infty$ sections over $Y$ of infinite dimensional bundles
which are still denoted by $H^\infty$. For $y\in Y$, $H^\infty_y$ is
the set of $C^\infty$ sections over $Z_y$ of $F'\otimes V$. For each
$y\in Y$, $$(D^{Z}\otimes V)_y\in \mathrm{End}(H^\infty_{y})$$ is
the twisted Dirac operator on the fiber $Z_y$.

The family odd signature operator is a twisted
family Dirac operator. Actually, we have $B^Z_{sig}=D^{Z}\otimes
F'$ (c.f. \cite{KL}).

As a higher analogue of the determinant line bundle, Lott (\cite{Lo})
constructs the index gerbe $\mathcal{G}^{D^Z\otimes V}$ with connection on
$Y$ for the family twisted odd Dirac operator $D^Z\otimes V$, the
curvature $R^{\mathcal{G}_{D\otimes V}}$ (a closed 3-form on $Y$) of
which  is equal to the the three-form component of the Atiyah-Singer
families index theorem. As remarked in \cite{Lo},  the curvature of
the index gerbe is certain nonabelian gauge anomaly in physics
(\cite{Fa}, cf. \cite{Lo}).

We have the following anomaly cancellation formulas for index gerbes. 

\begin{theorem}If fiber is $8m+1$ dimensional, then the following anomaly cancellation formula holds,
\be
R^{\mathcal{G}_{B^Z_{sig}}}-8\sum_{r=0}^{m}2^{6m-6r}R^{\mathcal{G}_{D\otimes
b_r(T_\CC Z)}}=0.\ee

\end{theorem}

\begin{theorem} If the fiber is $8m-3$ dimensional, then the following anomaly cancellation formula holds,
\be
R^{\mathcal{G}_{B^Z_{sig}}}-\sum_{r=0}^{m}2^{6m-6r}R^{\mathcal{G}_{D\otimes
z_r(T_\CC Z)}}=0.\ee\end{theorem}

If $\omega$ is a closed differential form on $Y$, denote the
cohomology class $\omega$ represents in the de Rham cohomology of
$Y$ by $[\omega]$.

We have the following cancellation formulas for cohomology anomalies.

\begin{theorem}If the fiber is $8m+1$ dimensional, then the following anomaly cancellation
 formula in cohomology holds,
\be
\sum_{r=0}^{m}2^{6m-6r}[R^{\mathcal{G}_{D\otimes b_r(T_\CC Z)}}]=0.\ee

\end{theorem}

\begin{theorem} If the fiber is $8m-3$ dimensional, then the following anomaly cancellation
formula in cohomology holds,
\be
\sum_{r=0}^{m}2^{6m-6r}[R^{\mathcal{G}_{D\otimes z_r(T_\CC Z)}}]=0.\ee

\end{theorem}

Putting $m=0$ and $m=1$ in the above theorems  and using (2.16) as well as (2.17), we have
\begin{corollary}If the fiber is $1$ dimensional, i.e. for the circle bundle case, the following anomaly
cancellation formula holds,
\be
R^{\mathcal{G}_{B^Z_{sig}}}+8R^{\mathcal{G}_{D^Z}}=0.\ee  The cohomology anomaly,
\be [R^{\mathcal{G}_{D^Z}}]=0. \ee

\end{corollary}

\begin{corollary}If the fiber is $5$ dimensional, then the following anomaly cancellation formula holds,
\be
R^{\mathcal{G}_{B^Z_{sig}}}+R^{\mathcal{G}_{D^Z\otimes T_\CC Z}}-21R^{\mathcal{G}_{D^Z}}=0.\ee
The cohomology anomaly,
\be [R^{\mathcal{G}_{D^Z\otimes T_\CC Z}}]-21[R^{\mathcal{G}_{D^Z}}]=0.\ee
\end{corollary}

\begin{corollary}If the fiber is $9$ dimensional, then the following anomaly cancellation formula holds,
\be
R^{\mathcal{G}_{B^Z_{sig}}}-8R^{\mathcal{G}_{D^Z\otimes T_\CC Z}}+8R^{\mathcal{G}_{D^Z}}=0.\ee
The cohomology anomaly,
\be [R^{\mathcal{G}_{D^Z\otimes T_\CC Z}}]-[R^{\mathcal{G}_{D^Z}}]=0.\ee

\end{corollary}

\subsection{Results for $\eta$-invariants} For $y\in Y$, let $\eta_y(D^{Z}\otimes V)(s)$ be the
eta function associated with $(D^{Z}\otimes V)_y$.  Define (\cite{APS1})
\be \overline{\eta}_y(D^{Z}\otimes V)(s) =\frac{\eta_y(D^{Z}\otimes V)(s)+\mathrm{ker}(D^{Z}\otimes V)_y}{2}.\ee  Denote 
$\overline{\eta}_y(D^{Z}\otimes V)(0) $ (a function on $Y$) by $\overline{\eta}(D^{Z}\otimes V)$.

We still adopt the setting of family odd signature operators and
family twisted Dirac operators on a family of odd manifolds in Section 2.3.  We
have the following theorems on the reduced $\eta$-invariants.
\begin{theorem}If the fiber is $8m+3$ dimensional, then
\be \exp\left\{2\pi \sqrt{-1}\left(\overline{\eta}(B^{Z}_{sig})-8\sum_{r=0}^{m}2^{6m-6r}\overline{\eta}
(D^{Z}\otimes b_r(T_\CC Z))\right)\right\}\ee is a constant function on $Y$.
\end{theorem}

\begin{theorem} If the fiber is $8m-1$ dimensional, then
\be \exp\left\{2\pi \sqrt{-1}\left(\overline{\eta}(B^{Z}_{sig})-\sum_{r=0}^{m}2^{6m-6r}\overline{\eta}(D^{Z}\otimes z_r(T_\CC Z))\right)\right\}\ee is a constant function on $Y$.

\end{theorem}

Putting $m=0$ and $m=1$ in the above theorems  and using (2.16) as well as (2.17), we have
\begin{corollary}If the fiber is $3$ dimensional, then
\be  \exp\left\{2\pi \sqrt{-1}\left(\overline{\eta}(B^{Z}_{sig})+8\overline{\eta}(D^{Z})\right)\right\}\ee is a constant function on $Y$.
\end{corollary}

\begin{corollary}If the fiber is $7$ dimensional, then
\be \exp\left\{2\pi \sqrt{-1}\left(\overline{\eta}(B^{Z}_{sig})+\overline{\eta}(D^{Z}\otimes T_\CC Z)-23\overline{\eta}(D^{Z})\right)\right\}\ee is a constant function on $Y$.

\end{corollary}

\begin{corollary}If the fiber is $11$ dimensional, then
\be \exp\left\{2\pi \sqrt{-1} \left(\overline{\eta}(B^{Z}_{sig})-8\overline{\eta}(D^{Z}\otimes T_\CC Z)+24\overline{\eta}(D^{Z})\right)\right\}\ee is a constant function on $Y$.
\end{corollary}

\section{Proofs}
In this section, we prove the theorems stated in Section 2. 

\subsection{Preliminaries}
Let $$ SL_2(\mathbf{Z}):= \left\{\left.\left(\begin{array}{cc}
                                      a&b\\
                                      c&d
                                     \end{array}\right)\right|a,b,c,d\in\mathbf{Z},\ ad-bc=1
                                     \right\}
                                     $$
 as usual be the modular group. Let
$$S=\left(\begin{array}{cc}
      0&-1\\
      1&0
\end{array}\right), \ \ \  T=\left(\begin{array}{cc}
      1&1\\
      0&1
\end{array}\right)$$
be the two generators of $ SL_2(\mathbf{Z})$. Their actions on
$\mathbf{H}$ are given by
$$ S:\tau\rightarrow-\frac{1}{\tau}, \ \ \ T:\tau\rightarrow\tau+1.$$

Let
$$ \Gamma_0(2)=\left\{\left.\left(\begin{array}{cc}
a&b\\
c&d
\end{array}\right)\in SL_2(\mathbf{Z})\right|c\equiv0\ \ (\rm mod \ \ 2)\right\},$$

$$ \Gamma^0(2)=\left\{\left.\left(\begin{array}{cc}
a&b\\
c&d
\end{array}\right)\in SL_2(\mathbf{Z})\right|b\equiv0\ \ (\rm mod \ \ 2)\right\}$$
be the two modular subgroups of $SL_2(\mathbf{Z})$. It is known
that the generators of $\Gamma_0(2)$ are $T,ST^2ST$ and the generators
of $\Gamma^0(2)$ are $STS,T^2STS$.(cf. \cite{C}).

If we act theta-functions by $S$ and $T$, the theta functions obey
the following transformation laws (cf. \cite{C}), \be
\theta(v,\tau+1)=e^{\pi \sqrt{-1}\over 4}\theta(v,\tau),\ \ \
\theta\left(v,-{1}/{\tau}\right)={1\over\sqrt{-1}}\left({\tau\over
\sqrt{-1}}\right)^{1/2} e^{\pi\sqrt{-1}\tau v^2}\theta\left(\tau
v,\tau\right)\ ;\ee \be \theta_1(v,\tau+1)=e^{\pi \sqrt{-1}\over
4}\theta_1(v,\tau),\ \ \
\theta_1\left(v,-{1}/{\tau}\right)=\left({\tau\over
\sqrt{-1}}\right)^{1/2} e^{\pi\sqrt{-1}\tau v^2}\theta_2(\tau
v,\tau)\ ;\ee \be\theta_2(v,\tau+1)=\theta_3(v,\tau),\ \ \
\theta_2\left(v,-{1}/{\tau}\right)=\left({\tau\over
\sqrt{-1}}\right)^{1/2} e^{\pi\sqrt{-1}\tau v^2}\theta_1(\tau
v,\tau)\ ;\ee \be\theta_3(v,\tau+1)=\theta_2(v,\tau),\ \ \
\theta_3\left(v,-{1}/{\tau}\right)=\left({\tau\over
\sqrt{-1}}\right)^{1/2} e^{\pi\sqrt{-1}\tau v^2}\theta_3(\tau
v,\tau)\ .\ee

\begin{definition} Let $\Gamma$ be a subgroup of $SL_2(\mathbf{Z}).$ A modular form over $\Gamma$ is a holomorphic function $f(\tau)$ on $\mathbf{H}\cup
\{\infty\}$ such that for any
 $$ g=\left(\begin{array}{cc}
             a&b\\
             c&d
             \end{array}\right)\in\Gamma\ ,$$
 the following property holds
 $$f(g\tau):=f(\frac{a\tau+b}{c\tau+d})=\chi(g)(c\tau+d)^lf(\tau), $$
 where $\chi:\Gamma\rightarrow\mathbf{C}^*$ is a character of
 $\Gamma$ and $l$ is called the weight of $f$.
 \end{definition}

If $\Gamma$ is a modular subgroup, let
$\mathcal{M}_\mathbf{R}(\Gamma)$ denote the ring of modular forms
over $\Gamma$ with real Fourier coefficients. Writing simply
$\theta_j=\theta_j(0,\tau),\ 1\leq j \leq 3,$ we introduce four (the second two have already appeared in Section 2.1)
explicit modular forms (cf. \cite{Lan}, \cite{Liu1}),
$$ \delta_1(\tau)=\frac{1}{8}(\theta_2^4+\theta_3^4), \ \ \ \
\varepsilon_1(\tau)=\frac{1}{16}\theta_2^4 \theta_3^4\ ,$$
$$\delta_2(\tau)=-\frac{1}{8}(\theta_1^4+\theta_3^4), \ \ \ \
\varepsilon_2(\tau)=\frac{1}{16}\theta_1^4 \theta_3^4\ ,$$
They
have the following Fourier expansions in $q^{1/2}$:
$$\delta_1(\tau)={1\over 4}+6q+6q^2+\cdots,\ \ \ \ \varepsilon_1(\tau)={1\over
16}-q+7q^2+\cdots\ , $$
$$\delta_2(\tau)=-{1\over 8}-3q^{1/2}-3q+\cdots,\ \ \ \
\varepsilon_2(\tau)=q^{1/2}+8q+\cdots\ ,$$ where the
\textquotedblleft $\cdots$" terms are the higher degree terms, all
of which have integral coefficients. They also satisfy the
transformation laws (cf. \cite{Lan}, \cite{Liu1}), \be
\delta_2\left(-\frac{1}{\tau}\right)=\tau^2\delta_1(\tau)\ \ \ \ \ ,
\ \ \ \ \
\varepsilon_2\left(-\frac{1}{\tau}\right)=\tau^4\varepsilon_1(\tau).\ee

Let
$\widehat{A}(TZ, \nabla^{Z})$ and $L(TZ, \nabla^{Z})$ be the
Hirzebruch characteristic forms defined respectively by (cf.
\cite{Z}) for $(TZ, \nabla^Z)$:
\begin{equation}
\begin{split}
&\widehat{A}(TZ, \nabla^{Z}) ={\det}^{1/2}\left({{\sqrt{-1}\over
4\pi}R^{Z} \over \sinh\left({ \sqrt{-1}\over
4\pi}R^{Z}\right)}\right), \\ &\widehat{L}(TZ, \nabla^{Z})
={\det}^{1/2}\left({{\sqrt{-1}\over 2\pi}R^{Z} \over \tanh\left({
\sqrt{-1}\over 4\pi}R^{Z}\right)}\right).
\end{split}
\end{equation}

If $\omega$ is a differential form, denote the $j$-component of $\omega$ by $\omega^{(j)}$.

\subsection{Proofs of Theorem 2.2.1, 2.2.3, 2.2.5, 2.3.1, 2.3.3 and 2.4.1.} Suppose the dimension of $TZ$ be $8m+1, 8m+2$ or $8m+3$.  For the vertical tangent bundle $TZ$, set \be P_1(\nabla^{Z}, \tau):=\left\{ \widehat{L}(TZ,\nabla^{Z})
\mathrm{ch}\left(\Theta_1(T_\CC
Z),\nabla^{\Theta_1(T_\CC Z)}\right)\right\}^{(8m+4)}\ee and \be
P_2(\nabla^{Z}, \tau):= \left\{ \widehat{A}(TZ,\nabla^{Z})
\mathrm{ch}\left(\Theta_2(T_\CC Z),\nabla^{\Theta_2(T_\CC
Z)}\right)\right\}^{(8m+4)}. \ee

\begin{proposition} $P_1(\nabla^{Z}, \tau)$ is a modular form of weight $4m+2$ over $\Gamma_0(2)$; $P_2(\nabla^{Z}, \tau)$ is a modular form of weight $4m+2$ over $\Gamma^0(2)$.
\end{proposition}

\begin{proof}  In terms of the theta functions, by the Chern-weil theory, the following identities hold,
\be P_1(\nabla^{Z}, \tau)=\left\{ \mathrm{det}^{1\over
2}\left(\frac{R^{Z}}{2{\pi}^2}\frac{\theta'(0,\tau)}{\theta(\frac{R^{Z}}{2{\pi}^2},\tau)}
\frac{\theta_{1}(\frac{R^{Z}}{2{\pi}^2},\tau)}{\theta_{1}(0,\tau)}\right)\right\}^{(8m+4)},
\ee

\be P_2(\nabla^{Z}, \tau)=\left\{\mathrm{det}^{1\over
2}\left(\frac{R^{Z}}{4{\pi}^2}\frac{\theta'(0,\tau)}{\theta(\frac{R^{Z}}{4{\pi}^2},\tau)}
\frac{\theta_{2}(\frac{R^{Z}}{4{\pi}^2},\tau)}{\theta_{2}(0,\tau)}\right)\right\}^{(8m+4)}.
\ee

Applying the transformation laws of the theta functions, we have
\be P_1\left(\nabla^{Z}, -\frac{1}{\tau}\right)=2^{4m+2}\tau^{4m+2}P_2(\nabla^{Z}, \tau), \ P_1(\nabla^{Z}, \tau+1)=P_1(\nabla^{Z}, \tau).  \ee Because the generators of $\Gamma_0(2)$ are $T,ST^2ST$ and the generators
of $\Gamma^0(2)$ are $STS,T^2STS$, the proposition follows easily.
\end{proof}

\begin{lemma} [\protect cf. \cite{Liu1}] One has that $\delta_1(\tau)\ (resp.\ \varepsilon_1(\tau) ) $
is a modular form of weight $2 \ (resp.\ 4)$ over $\Gamma_0(2)$,
$\delta_2(\tau) \ (resp.\ \varepsilon_2(\tau))$ is a modular form
of weight $2\ (resp.\ 4)$ over $\Gamma^0(2)$, while
$\delta_3(\tau) \ (resp.\ \varepsilon_3(\tau))$ is a modular form
of weight $2\ (resp.\ 4)$ over $\Gamma_\theta(2)$ and moreover
$\mathcal{M}_\mathbf{R}(\Gamma^0(2))=\mathbf{R}[\delta_2(\tau),
\varepsilon_2(\tau)]$.
\end{lemma}

We then apply Lemma 3.1 to $P_2(\nabla^Z, \tau)$ to get that \be
P_2(\nabla^Z, \tau)
 =h_0(T_\CC Z)(8\delta_2)^{2m+1}+h_1(T_\CC Z)(8\delta_2)^{2m-1}\varepsilon_2
+\cdots+h_m(T_\CC Z)(8\delta_2)\varepsilon_2^m. \ee Comparing (2.14), we can see that
$$h_r(T_\CC Z)=\left\{ \widehat{A}(TZ,\nabla^{Z})
\mathrm{ch}(b_r(T_\CC Z))\right\}^{(8m+4)}, \ 0\leq r \leq m. $$

By (3.5), (3.11) and (3.12), we have
\be P_1(\nabla^Z, \tau)=2^{4m+2}[h_0(T_\CC Z)(8\delta_1)^{2m+1}+h_1(T_\CC Z)(8\delta_1)^{2m-1}\varepsilon_1
+\cdots+h_m(T_\CC Z)(8\delta_1)\varepsilon_1^m]. \ee

Comparing the constant term of the above equality, we see that
\be \{\widehat{L}(TZ, \nabla^Z)\}^{(8m+4)}=8\sum_{r=0}^{m}2^{6m-6r}\{\widehat{A}(TZ,\nabla^{Z})
\mathrm{ch}(b_r(T_\CC Z))\}^{(8m+4)}. \ee

In the following, we will deal with the even case and odd case respectively.

\subsubsection{The case of  even dimensional fibers}

We have the following Bismut-Freed theorem on the curvature of the determinant line bundle with Bismut-Freed connection. 
\begin{theorem}[\protect Bismut-Freed, \cite{BF2}]
\be R^{\mathcal{L}_{D^{Z}\otimes V}}=2\pi \sqrt{-1} \left\{\int_Z \widehat{A}(TZ,\nabla^{Z})
\mathrm{ch}(V, \nabla^V)\right\}^{(2)}. \ee \end{theorem}

To detect mod $k$ information of the first Chern class of the determinant line bundle, Freed has the following result.
 \begin{theorem} [\protect Freed, \cite{F}] If $(\Sigma, S)$ is a $\mathbf{Z}/k$ surface and  $f:\overline{\Sigma}\to Y$ is a map, then
\be
\begin{split} & c_1\left(f^*\left(\mathcal{L}_{D^Z\otimes V}\right)\right)[\overline{\Sigma}]\\
=& \frac{1}{k}\frac{\sqrt{-1}}{2\pi}\int_\Sigma f^*\left(\int_Z \widehat{A}(TZ,\nabla^{Z})
\mathrm{ch}(V, \nabla^V)\right)+\frac{\sqrt{-1}}{2\pi}\mathrm{ln
hol}_{\mathcal{L}_{D^Z\otimes V}}(S) \ \
\mathrm{mod}\  1,\end{split} \ee where we view
$\mathbf{Z}/k\cong\mathbf{Z}\left[1/k \right]/\mathbf{Z}\subset
\mathbf{Q}/\mathbf{Z}.$
\end{theorem}

If $TZ$ is of dimension $8m+2$, integrating both sides of (3.14)
along the fiber, we have \be \left\{\int_Z \widehat{L}(TZ,
\nabla^Z)\right\}^{(2)}-8\sum_{r=0}^{m}2^{6m-6r}\left\{\int_Z
\widehat{A}(TZ,\nabla^{Z}) \mathrm{ch}(b_r(T_\CC
Z))\right\}^{(2)}=0.\ee

By Theorem 3.2.1, we get
\be
\begin{split} &R^{\mathcal{L}_{D^Z_{sig}}}-8\sum_{r=0}^{m}2^{6m-6r}R^{\mathcal{L}_{D^Z\otimes
b_r(T_\CC Z)}}\\
=&2\pi \sqrt{-1}\left\{\int_Z \widehat{L}(TZ, \nabla^Z)\right\}^{(2)}-8\sum_{r=0}^{m}2^{6m-6r}2\pi \sqrt{-1}\left\{\int_Z \widehat{A}(TZ,\nabla^{Z})
\mathrm{ch}(b_r(T_\CC Z))\right\}^{(2)}\\
=&0. \end{split} \ee  Therefore Theorem 2.2.1 follows.

Similarly, Freed's Theorem 3.2.2 and (3.17) give us
\be
\begin{split}&c_1(f^*(\mathcal{L}_{D^Z_{sig}}))[\overline{\Sigma}]
-8\sum_{r=0}^{m}2^{6m-6r}c_1\left(f^*\left(\mathcal{L}_{D^Z\otimes
b_r(T_\CC Z)}\right)\right)[\overline{\Sigma}]\\
-&\left(\frac{\sqrt{-1}}{2\pi}\mathrm{ln
hol}_{\mathcal{L}_{D^Z_{sig}}}(S)-8\sum_{r=0}^{m}2^{6m-6r}\frac{\sqrt{-1}}{2\pi}\mathrm{ln
hol}_{\mathcal{L}_{D^Z\otimes
b_r(T_\CC Z)}}(S)\right)\equiv 0\ \
\mathrm{mod}\  1\end{split}\ee and so Theorem 2.2.3 follows.

To prove Theorem 2.2.5,  it's not hard to see from (32), (38) and (56) in \cite{AW} that, up to a same constant,
$$ \widehat{I}_{1/2}=\left\{\int_Z \widehat{A}(TZ, \nabla^Z)\right\}^{(2)}=R^{\mathcal{L}_{D^Z}},$$
$$\widehat{I}_{3/2}=\left\{\int_Z \widehat{A}(TZ, \nabla^Z)(\mathrm{ch}(T_\CC Z, \nabla^Z)-1)\right\}^{(2)}=R^{\mathcal{L}_{D^Z\otimes T_\CC Z}} -R^{\mathcal{L}_{D^Z}}$$ and
$$ \widehat{I}_{A}=-\frac{1}{8}\left\{\int_Z\widehat{L}(TZ, \nabla^Z)\right\}^{(2)}=-\frac{1}{8}R^{\mathcal{L}_{D^Z_{sig}}},$$
where in the fiber bundle $Z\rightarrow M  \rightarrow Y,$ $Z$ is a $4k+2$ dimensional spin manifold and $Y$ is the quotient space of the space of metrics on $Z$ by the action of certain subgroup of $\mathrm{Diff}(M)$.

In dimension 2, by (2.24),
$$ -\widehat{I}_{1/2}+\widehat{I}_{A}=-R^{\mathcal{L}_{D^Z}}-\frac{1}{8}R^{\mathcal{L}_{D^Z_{sig}}}=-\frac{1}{8}(R^{\mathcal{L}_{D^Z_{sig}}}+8R^{\mathcal{L}_{D^Z}})=0.$$ Therefore (1.1) follows.

In dimension 6, by (2.26),
\begin{equation*}\begin{split} &21\widehat{I}_{1/2}-\widehat{I}_{3/2}+8\widehat{I}_{A}\\
=&21 R^{\mathcal{L}_{D^Z}}-(R^{\mathcal{L}_{D^Z\otimes T_\CC Z}} -R^{\mathcal{L}_{D^Z}})-R^{\mathcal{L}_{D^Z_{sig}}}\\
=& 22R^{\mathcal{L}_{D^Z}}-R^{\mathcal{L}_{D^Z\otimes T_\CC Z}}-R^{\mathcal{L}_{D^Z_{sig}}}\\
=&0.\end{split}
\end{equation*} Therefore (1.2) follows.

In dimension 10, by (2.28),
\begin{equation*}\begin{split} &-\widehat{I}_{1/2}+\widehat{I}_{3/2}+\widehat{I}_{A}\\
=&-R^{\mathcal{L}_{D^Z}}+(R^{\mathcal{L}_{D^Z\otimes T_\CC Z}} -R^{\mathcal{L}_{D^Z}})-\frac{1}{8}R^{\mathcal{L}_{D^Z_{sig}}}\\
=& -\frac{1}{8}(16R^{\mathcal{L}_{D^Z}}-8R^{\mathcal{L}_{D^Z\otimes T_\CC Z}}+R^{\mathcal{L}_{D^Z_{sig}}})\\
=&0.\end{split}
\end{equation*} Therefore (1.3) follows.

\subsubsection{The case of odd dimensional fibers }

Lott has the following theorem for the curvature of index gerbes.
\begin{theorem} [\protect Lott, \cite{Lo}]
\be R^{\mathcal{G}_{D^Z\otimes V}}=\left\{\int_Z \widehat{A}(TZ,\nabla^{Z})
\mathrm{ch}(V, \nabla^V)\right\}^{(3)}. \ee \end{theorem}

If $TZ$ is of dimension $8m+1$, integrating both sides of (3.14)
along the fiber, we get \be \left\{\int_Z \widehat{L}(TZ,
\nabla^Z)\right\}^{(3)}-8\sum_{r=0}^{m}2^{6m-6r}\left\{\int_Z
\widehat{A}(TZ,\nabla^{Z}) \mathrm{ch}(b_r(T_\CC
Z))\right\}^{(3)}=0.\ee

Note that we have $\widehat{A}(TZ,\nabla^{Z})
\mathrm{ch}(F', \nabla^{F'})=\widehat{L}(TZ, \nabla^Z)$ (\cite{KL}).
So by Theorem 3.2.3 and (3.21), we have
\be
\begin{split} &R^{\mathcal{G}_{B^Z_{sig}}}-8\sum_{r=0}^{m}2^{6m-6r}R^{\mathcal{G}_{D^Z\otimes
b_r(T_\CC Z)}}\\
=&\left\{\int_Z \widehat{L}(TZ, \nabla^Z)\right\}^{(3)}-8\sum_{r=0}^{m}2^{6m-6r}\left\{\int_Z \widehat{A}(TZ,\nabla^{Z})
\mathrm{ch}(b_r(T_\CC Z))\right\}^{(3)}\\
=&0. \end{split} \ee  Therefore Theorem 2.3.1 follows.

On the family odd signature operators, there is the following theorem:
\begin{theorem}[\protect Ebert, \cite{Eb}] The family index of the odd signature operator on an oriented bundle $M \rightarrow Y$ with odd dimensional fibers is trivial, i.e, $\mathrm{ind}(B^Z_{sig})=0\in K^1(Y)$.
\end{theorem}

The following theorem on the odd Chern form for a family of self-adjoint Dirac operators is due to Bismut and Freed.
\begin{theorem} [\protect Bismut-Freed, \cite{BF2}] $\int_Z \widehat{A}(TZ,\nabla^{Z})
\mathrm{ch}(V, \nabla^V)$ represents the odd Chern character of $\mathrm{ind}(D^Z\otimes V)$.

\end{theorem}

Combining Theorem 3.2.4 and 3.2.5, we see that $[\int_Z \widehat{L}(TZ,\nabla^{Z})]$ is zero in de Rham cohomology. In particular, by Theorem 3.2.3, $\left[R^{\mathcal{G}_{B^Z_{sig}}}\right]=0$.  Therefore, (3.22) implies that
\be \sum_{r=0}^{m}2^{6m-6r}[R^{\mathcal{G}_{D^Z\otimes
b_r(T_\CC Z)}}]=0.\ee So Theorem 2.3.3 follows.

If $d$ is a real number,  let $\{d\}$ denote the image of $d$ in $\mathbf{R}/\mathbf{Z}$. As noted in \cite{APS1, APS3}, $\overline{\eta}_y(D^{Z}\otimes V)(0)$ has integer jumps and therefore $\{\overline{\eta}((D^{Z}\otimes V))\}$ is a $C^\infty$ function of on Y with values in $\mathbf{R}/\mathbf{Z}$ (\cite{APS1, APS3}). For odd dimensional fibers,  we have the following Bismut-Freed theorem for the reduced $\eta$-invariants.
\begin{theorem} [\protect Bismut-Freed, \cite{BF2}]
\be d\{\overline{\eta}(D^{Z}\otimes V)\}=\left\{\int_Z \widehat{A}(TZ,\nabla^{Z})
\mathrm{ch}(V, \nabla^V)\right\}^{(1)}. \ee \end{theorem}

If $TZ$ is of $8m+3$ dimensional, integrating both sides of (3.14) along the fiber, we get
\be \left\{\int_Z \widehat{L}(TZ, \nabla^Z)\right\}^{(1)}-8\sum_{r=0}^{m}2^{6m-6r}\left\{\int_Z \widehat{A}(TZ,\nabla^{Z})
\mathrm{ch}(b_r(T_\CC Z))\right\}^{(1)}=0.\ee Then the Bismut-Freed Theorem 3.2.6 gives us
\be \begin{split} &d\{\overline{\eta}(B^{Z}_{sig})\}-8\sum_{r=0}^{m}2^{6m-6r}d\{\overline{\eta}(D^{Z}\otimes b_r(T_\CC Z))\}\\
=&\left\{\int_Z \widehat{L}(TZ, \nabla^Z)\right\}^{(1)}-8\sum_{r=0}^{m}2^{6m-6r}\left\{\int_Z \widehat{A}(TZ,\nabla^{Z})
\mathrm{ch}(b_r(T_\CC Z))\right\}^{(1)}\\
=&0. \end{split}\ee Therefore we obtain
\be d\left(\{\overline{\eta}(B^{Z}_{sig})\}-8\sum_{r=0}^{m}2^{6m-6r}\{\overline{\eta}(D^{Z}\otimes b_r(T_\CC Z))\}\right)=0. \ee Since $Y$ is connected,
$$\{\overline{\eta}(B^{Z}_{sig})\}-8\sum_{r=0}^{m}2^{6m-6r}\{\overline{\eta}(D^{Z}\otimes b_r(T_\CC Z))\}$$ must be a constant function on $Y$.  Therefore  it's not hard to see that Theorem 2.4.1 follows.

\subsection{Proofs of Theorem 2.2.2, 2.2.4, 2.3.2, 2.3.4 and 2.4.2} The proofs are similar to the proofs of Theorem 2.2.1. 2.2.3, 2.3.1, 2.3.3 and 2.4.1.

Let the dimension of $TZ$ be $8m-1, 8m-2$ or $8m-3$. For the vertical tangent bundle $TZ$, set \be Q_1(\nabla^{Z}, \tau):=\left\{ \widehat{L}(TZ,\nabla^{Z})
\mathrm{ch}\left(\Theta_1(T_\CC
Z),\nabla^{\Theta_1(T_\CC Z)}\right)\right\}^{(8m)},\ee  \be
Q_2(\nabla^{Z}, \tau):= \left\{ \widehat{A}(TZ,\nabla^{Z})
\mathrm{ch}\left(\Theta_2(T_\CC Z),\nabla^{\Theta_2(T_\CC
Z)}\right)\right\}^{(8m)}. \ee

Similar to Proposition 3.1, we have
\be Q_1(\nabla^{Z}, \tau)=\left\{ \mathrm{det}^{1\over
2}\left(\frac{R^{Z}}{2{\pi}^2}\frac{\theta'(0,\tau)}{\theta(\frac{R^{Z}}{2{\pi}^2},\tau)}
\frac{\theta_{1}(\frac{R^{Z}}{2{\pi}^2},\tau)}{\theta_{1}(0,\tau)}\right)\right\}^{(8m)},
\ee
\be Q_2(\nabla^{Z}, \tau)=\left\{\mathrm{det}^{1\over
2}\left(\frac{R^{Z}}{4{\pi}^2}\frac{\theta'(0,\tau)}{\theta(\frac{R^{Z}}{4{\pi}^2},\tau)}
\frac{\theta_{2}(\frac{R^{Z}}{4{\pi}^2},\tau)}{\theta_{2}(0,\tau)}\right)\right\}^{(8m)}.
\ee

Also $Q_1(\nabla^{Z}, \tau)$ is a modular form of weight $4m$ over $\Gamma_0(2)$ and $Q_2(\nabla^{Z}, \tau)$ is a modular form of weight $4m$ over $\Gamma^0(2)$. Moreover,
\be Q_1\left(\nabla^{Z}, -\frac{1}{\tau}\right)=2^{4m}\tau^{4m}Q_2(\nabla^{Z}, \tau), \ Q_1(\nabla^{Z}, \tau+1)=Q_1(\nabla^{Z}, \tau).  \ee

Similar to (3.12) and (3.13), by using Lemma 3.1 and (3.32), we have
\be \begin{split} Q_2(\nabla^Z, \tau)
 =&\left\{ \widehat{A}(TZ,\nabla^{Z})
\mathrm{ch}(z_0(T_\CC Z))\right\}^{(8m)}(8\delta_2)^{2m}\\
&+\left\{ \widehat{A}(TZ,\nabla^{Z})
\mathrm{ch}(z_1(T_\CC Z))\right\}^{(8m)}(8\delta_2)^{2m-2}\varepsilon_2\\
&+\cdots+\left\{ \widehat{A}(TZ,\nabla^{Z})
\mathrm{ch}(z_m(T_\CC Z))\right\}^{(8m)}\varepsilon_2^m, \end{split}\ee and
\be \begin{split} Q_1(\nabla^Z, \tau)=&2^{4m}\left[\left\{ \widehat{A}(TZ,\nabla^{Z})
\mathrm{ch}(z_0(T_\CC Z))\right\}^{(8m)}(8\delta_1)^{2m}\right.\\
&\left.+\cdots+\left\{ \widehat{A}(TZ,\nabla^{Z})
\mathrm{ch}(z_m(T_\CC Z))\right\}^{(8m)}\varepsilon_1^m\right]. \end{split}\ee

Comparing the constant term of the above equality, we see that
\be \{\widehat{L}(TZ, \nabla^Z)\}^{(8m)}=\sum_{r=0}^{m}2^{6m-6r}\{\widehat{A}(TZ,\nabla^{Z})
\mathrm{ch}(z_r(T_\CC Z))\}^{(8m)}. \ee

Then one can integrate both sides of (3.35) along the fiber and combine the theorems of Bismut-Freed, Freed, Lott and Ebert to obtain Theorem 2.2.2, 2.2.4, 2.3.2, 2.3.4 and 2.4.2.

$$ $$

\noindent {\bf Acknowledgements} We are grateful to Professor Weiping Zhang for helpful discussions.  The first author is supported by a start-up grant from National University of Singapore. 

$$ $$

\end{document}